\documentstyle[12pt]{article}

\font\twlgot =eufm10 scaled \magstep1 \font\egtgot =eufm8
\font\sevgot =eufm7 \font\twlmsb =msbm10 scaled \magstep1
\font\egtmsb =msbm8 \font\sevmsb =msbm7

\newfam\gotfam
\def\pgot{\fam\gotfam\twlgot}
\textfont\gotfam\twlgot \scriptfont\gotfam\egtgot
\scriptscriptfont\gotfam\sevgot
\def\got{\protect\pgot}
\newfam\msbfam
\textfont\msbfam\twlmsb \scriptfont\msbfam\egtmsb
\scriptscriptfont\msbfam\sevmsb
\def\Bbb{\protect\pBbb}
\def\pBbb{\relax\ifmmode\expandafter\Bb\else\typeout{You cann't use
Bbb in text mode}\fi}
\def\Bb #1{{\fam\msbfam\relax#1}}

\newcommand{\gd}{{\got d}}

\def\op#1{\mathop{\fam0 #1}\limits}

\newcommand{\id}{{\rm Id\,}}
\newcommand{\Ker}{{\rm Ker\,}}

\newcommand{\beq}{\begin{equation}}
\newcommand{\eeq}{\end{equation}}
\newcommand{\ben}{\begin{eqnarray}}
\newcommand{\een}{\end{eqnarray}}
\newcommand{\be}{\begin{eqnarray*}}
\newcommand{\ee}{\end{eqnarray*}}
\newcommand{\bea}{\begin{eqalph}}
\newcommand{\eea}{\end{eqalph}}
\newcommand{\cA}{{\cal A}}

\newcommand{\cD}{{\cal D}}

\newcommand{\cJ}{{\cal J}}

\newcommand{\cQ}{{\cal Q}}
\newcommand{\cF}{{\cal F}}

\newcommand{\cO}{{\cal O}}

\newcommand{\cK}{{\cal K}}

\newcommand{\bL}{{\bf L}}

\newcommand{\bb}{{\bf 1}}
\newcommand{\hm}{{\rm Hom\,}}
\newcommand{\bll}{\bullet}
\newcommand{\dif}{{\rm Diff\,}}
\newcommand{\gf}{{\got f}}

\newcommand{\dl}{\delta}

\newcommand{\f}{\phi}

\newcommand{\m}{\mu}

\newcommand{\g}{\gamma}
\newcommand{\G}{\Gamma}

\newcommand{\di}{{\rm dim\,}}

\newcommand{\e}{\epsilon}

\newcommand{\si}{\sigma}

\newcommand{\w}{\wedge}
\newcommand{\wt}{\widetilde}

\newcommand{\ol}{\overline}
\newcommand{\dr}{\partial}
\newcommand{\ar}{\op\longrightarrow}

\newcommand{\ot}{\otimes}

\newcommand{\ve}{\varepsilon}

\newcounter{eqalph}
\newcounter{equationa}
\newcounter{remark}
\newcounter{example}
\newcounter{theorem}
\newcounter{proposition}
\newcounter{lemma}
\newcounter{corollary}
\newcounter{definition}
\newenvironment{eqalph}{\stepcounter{equation}
\setcounter{equationa}{\value{equation}} \setcounter{equation}{0}

\begin{eqnarray}}{\end{eqnarray}\setcounter{equation}{\value{equationa}}}

\def\theremark{\arabic{remark}}

\def\thetheorem{\arabic{theorem}}

\newenvironment{proof}{
{\it Proof:}}{}

\newenvironment{ex}{\refstepcounter{remark}{\it Example
\theremark:}}{}
\newenvironment{theo}{\refstepcounter{theorem}
{\bf Theorem \thetheorem:}}{}

\newenvironment{lem}{\refstepcounter{theorem}
{\bf Lemma \thetheorem:}}{}

\newenvironment{defi}{\refstepcounter{theorem}
{\bf Definition \thetheorem:}}{}

\newcommand{\mar}[1]{}

\begin{document}
\hbox{}

{\parindent=0pt

{\large\bf Differential operators on Schwartz distributions. Jet
formalism}
\bigskip

{\sc G. Sardanashvily}


{\sl Department of Theoretical Physics, Moscow State University,
Russia}

\bigskip
\bigskip

{\small Differential operators on Schwartz distributions
conventionally are defined as the transpose of differential
operators on functions with compact support. They do not exhaust
all differential operators. We follow algebraic formalism of
differential operators on modules over commutative rings. In a
general setting, Schwartz distributions on sections with compact
support of vector bundles over an arbitrary smooth manifold are
considered.}

}

\section{Introduction}

Quantum field theory provides examples of differential operators
and differential equations on distributions.

Let $U$ be an open subset of $\Bbb R^n$, $\cD(U)$ the space of
smooth real functions with compact support, and $\cD(U)'$ the
space of Schwartz distribution on $\cD(U)$. Differential operators
on $\cD(U)'$ conventionally are defined as the transpose of
differential operators on $\cD(U)$ \cite{dal}. However, they do
not exhaust all differential operators on $\cD(U)'$.

We follow algebraic formalism of differential operators on modules
over commutative rings and jets of modules (Section 2)
\cite{book97,kras}.

In a general setting, Schwartz distributions on sections with
compact support of a vector bundles $Y$ on an arbitrary smooth
manifold $X$ are considered. We follow familiar formalism of
distributions, not the nonlinear ones \cite{kunz}.

Let $Y(X)$ denote a $C^\infty(X)$-module of global sections of
$Y\to X$. Let $E\to X$ be a vector bundle. The $E(X)$-valued
differential operators defined on $Y(X)$ as a $C^\infty(X)$-module
coincide with the familiar ones (Section 3). They constitute the
$C^\infty(X)$-module (\ref{141}).

Let $\cD(Y)$ be a $C^\infty(X)$-module of sections with compact
support of $Y\to X$. It is provided with a $LF$-topology similar
to that on $\cD(U)$. With this topology, $\cD(U)$ is a nuclear
vector space and a topological $C^\infty(X)$-module. Differential
operators on $\cD(Y)$ as a $C^\infty(X)$-module are defined
(Section 4). We show that that there is one-to-one correspondence
between $E(X)$-valued differential operators on $\cD(Y)$ and those
on $Y(X)$ (Theorem \ref{130}), and that $\cD(Y)$-valued
differential operators on $\cD(Y)$ are continuous (Theorem
\ref{106}).

Let $\cD(Y)'$ be the topological dual of $\cD(Y)$ endowed with the
strong topology. It is a nuclear vector space whose topological
dual is $\cD(Y)$, and it is a topological $C^\infty(X)$-module.
Differential operators on $\cD(Y)'$ as a $C^\infty(X)$-module are
defined. In particular, the transpose of any $\cD(Y)$-valued
differential operator on $\cD(Y)$ is a differential operator on
$\cD(Y)'$ (Theorem \ref{205}). However, a differential operator on
$\cD(Y')$ need not be of this type. We show that a
$\cD(Y)'$-valued differential operator on $\cD(Y)'$ is the
transpose of a $\cD(Y)$-valued differential operator on $\cD(Y)$
iff it is continuous (Theorem \ref{220}).

\section{Differential operators on modules}

This Section summarizes the relevant material on differential
operator on modules over a commutative ring
\cite{book97,kras,sard09}.

Let $\cK$ be a commutative ring (i.e., a commutative unital
algebra) and $\cA$ a commutative $\cK$-ring. Let $P$ and $Q$ be
$\cA$-modules. The $\cK$-module $\hm_\cK (P,Q)$ of
$\cK$-homomorphisms $\Phi:P\to Q$ can be endowed with the two
different $\cA$-module structures
\mar{5.29}\beq
(a\Phi)(p)= a\Phi(p),  \qquad  (a\bll\Phi)(p) = \Phi (a p),\qquad
a\in \cA, \quad p\in P. \label{5.29}
\eeq
We refer to the second one as a $\cA^\bll$-module structure. Let
us put
\mar{spr172}\beq
\dl_a\Phi= a\Phi -a\bll\Phi, \qquad a\in\cA. \label{spr172}
\eeq

\begin{defi} \label{ws131} \mar{ws131}
An element $\Delta\in\hm_\cK(P,Q)$ is called a $k$-order
$Q$-valued differential operator on $P$ if
\be
\dl_{a_0}\circ\cdots\circ\dl_{a_k}\Delta=0
\ee
for any tuple of $k+1$ elements $a_0,\ldots,a_k$ of $\cA$. The set
$\dif_k(P,Q)$ of these operators inherits the $\cA$- and
$\cA^\bll$-module structures (\ref{5.29}).
\end{defi}

In particular, zero order differential operators are
$\cA$-homomorphisms $P\to Q$. A first order differential operator
$\Delta$ satisfies the condition
\be
\dl_b\circ\dl_a\,\Delta(p)= ba\Delta(p) -b\Delta(ap)
-a\Delta(bp)+\Delta(abp) =0, \quad a,b\in\cA.
\ee

Let $P=\cA$. Any zero order $Q$-valued differential operator
$\Delta$ on $\cA$ is defined by its value $\Delta(\bb)$. Then
there is an isomorphism
\be
\dif_0(\cA,Q)=Q
\ee
via the association
\be
Q\ni q\to \Delta_q\in \dif_0(\cA,Q), \qquad \Delta_q(\bb)=q,
\ee
A first order $Q$-valued differential operator $\Delta$ on $\cA$
fulfils the condition
\be
\Delta(ab)=b\Delta(a)+ a\Delta(b) -ba \Delta(\bb), \qquad
a,b\in\cA.
\ee
It is a $Q$-valued derivation of $\cA$ if $\Delta(\bb)=0$, i.e.,
the Leibniz rule
\mar{+a20}\beq
\Delta(ab) = \Delta(a)b + a\Delta(b), \qquad  a,b\in \cA,
\label{+a20}
\eeq
holds. Any first order differential operator on $\cA$ falls into
the sum
\be
\Delta(a)= a\Delta(\bb) +[\Delta(a)-a\Delta(\bb)]
\ee
of a zero order differential operator and a derivation.
Accordingly, there is an $\cA$-module decomposition
\mar{spr156'}\beq
\dif_1(\cA,Q) = Q \oplus\gd(\cA,Q), \label{spr156'}
\eeq
where $\gd(\cA,Q)$ is an $\cA$-module of $Q$-valued derivations of
$\cA$.

If $P=Q=\cA$, the derivation module $\gd\cA$ of $\cA$ is a Lie
$\cK$-algebra. Accordingly, the decomposition (\ref{spr156'})
takes the form
\mar{spr156}\beq
\dif_1(\cA) = \cA \oplus\gd\cA. \label{spr156}
\eeq

\begin{ex} \mar{115} \label{115}
Let $X$ be an $n$-dimensional real smooth  manifold coordinated by
$x^\m$, and let $C^\infty(X)$ be an $\Bbb R$-ring of smooth real
functions on $X$. There is one-to-one correspondence between the
derivations of $C^\infty(X)$ and the vector fields on $X$. It is
given by the expression
\mar{116}\beq
T(X)=TX(X)\ni u\leftrightarrow \bL_u\in \gd C^\infty(X), \qquad
\bL_u(f)=u^\m\dr_\m f, \qquad f\in C^\infty(X), \label{116}
\eeq
where $\bL_u$ denotes the Lie derivative along $u$.
\end{ex}

The study of $Q$-valued differential operators on an $\cA$-module
$P$ is reduced to that of $Q$-valued differential operators on a
ring $\cA$ as follows.

\begin{theo} \label{ws109} \mar{ws109}
Let us consider an $\cA$-homomorphism
\mar{n2}\beq
h_k: \dif_k(\cA,Q)\to Q, \qquad h_k(\Delta)=\Delta(\bb).
\label{n2}
\eeq
Any $k$-order $Q$-valued differential operator $\Delta\in
\dif_k(P,Q)$ on $P$ uniquely factorizes as
\mar{n13}\beq
\Delta:P\ar^{\gf_\Delta} \dif_k(\cA,Q)\ar^{h_k} Q \label{n13}
\eeq
through the homomorphism $h_k$ (\ref{n2}) and some homomorphism
\mar{n0}\beq
\gf_\Delta: P\to \dif_k(\cA,Q), \qquad (\gf_\Delta
p)(a)=\Delta(ap), \qquad a\in \cA, \label{n0}
\eeq
of an $\cA$-module $P$ to an $\cA^\bll$-module $\dif_k(\cA,Q)$.
The assignment $\Delta\to\gf_\Delta$ defines an
$\cA^\bll-\cA$-module isomorphism
\mar{n1}\beq
\dif_k(P,Q)=\hm_{\cA-\cA^\bll}(P,\dif_k(\cA,Q)). \label{n1}
\eeq
\end{theo}

In a different way, $k$-order differential operators on a module
$P$ are represented by zero order differential operators on a
module of $k$-order jets of $P$ as follows.

Given an $\cA$-module $P$, let us consider a tensor product
$\cA\otimes_\cK P$ of $\cK$-modules $\cA$ and $P$. We put
\mar{spr173}\beq
\dl^b(a\otimes p)= (ba)\otimes p - a\otimes (b p), \qquad p\in P,
\qquad a,b\in\cA.  \label{spr173}
\eeq
Let us denote by $\m^{k+1}$ a submodule of $\cA\ot_\cK P$
generated by elements of the type
\be
\dl^{b_0}\circ \cdots \circ\dl^{b_k}(a\otimes p).
\ee

\begin{defi} \mar{200} \label{200}
A $k$-order jet module $\cJ^k(P)$ of a module $P$ is the quotient
of the $\cK$-module $\cA\otimes_\cK P$ by $\m^{k+1}$. We denote
its elements $c\ot_kp$.
\end{defi}

In particular, a first order jet module $\cJ^1(P)$ is generated by
elements $\bb\ot_1 p$ modulo the relations
\mar{mos041}\beq
\dl^a\circ \dl^b(\bb\ot_1 p)= ab\otimes_1 p -b\otimes_1 (ap)
-a\otimes_1 (bp) +\bb\ot_1(abp) =0. \label{mos041}
\eeq

A $\cK$-module $\cJ^k(P)$ is endowed with the $\cA$- and
$\cA^\bll$-module structures
\mar{+a21}\beq
b(a\ot_k p)= ba\ot_k p, \qquad b\bll(a\otimes_k p)= a\otimes_k
(bp). \label{+a21}
\eeq
There exists a homomorphism
\mar{5.44}\beq
J^k: P\ni p\to \bb\otimes_k p\in \cJ^k(P) \label{5.44}
\eeq
of an $\cA$-module $P$ to an $\cA^\bll$-module $\cJ^k(P)$ such
that $\cJ^k(P)$, seen as an $\cA$-module, is generated by elements
$J^kp$, $p\in P$.

Due to the natural monomorphisms $\m^r\to \m^k$ for all $r>k$,
there are $\cA$-module epimorphisms of jet modules
\be
\pi^{i+1}_i: \cJ^{i+1}(P)\to \cJ^i(P).
\ee
In particular,
\mar{+a13}\beq
\pi^1_0:\cJ^1(P) \ni a\ot_1 p\to ap \in P.\label{+a13}
\eeq

\begin{theo} \label{t6} \mar{t6}
Any $k$-order $Q$-valued differential operator $\Delta$ on an
$\cA$-module $P$ factorizes uniquely
\be
\Delta: P\ar^{J^k} \cJ^k(P)\ar^{{\got f}^\Delta} Q
\ee
through the homomorphism $J^k$ (\ref{5.44}) and some
$\cA$-homomorphism ${\got f}^\Delta: \cJ^k(P)\to Q$. The
association $\Delta\to {\got f}^\Delta$ yields an
$(\cA^\bll-\cA)$-module isomorphism
\mar{5.50}\beq
\dif_k(P,Q)=\hm_{\cA}(\cJ^k(P),Q). \label{5.50}
\eeq
\end{theo}

Let us consider jet modules $\cJ^k=\cJ^k(\cA)$ of a ring $\cA$
itself. In particular, the first order jet module $\cJ^1$ consists
of the elements $a\otimes_1 b$, $a,b\in\cA$, subject to the
relations
\mar{5.53}\beq
ab\otimes_1 \bb -b\otimes_1 a -a\otimes_1 b +\bb\ot_1(ab) =0.
\label{5.53}
\eeq
The $\cA$- and $\cA^\bll$-module structures (\ref{+a21}) on
$\cJ^1$ read
\be
c(a\ot_1 b)=(ca)\ot_1 b,\qquad c\bll(a\ot_k b)=
a\ot_1(cb)=(a\ot_1b)c.
\ee

Theorems \ref{ws109} and \ref{t6} are completed with forthcoming
Theorem \ref{111} so that any one of them is a corollary of the
others.

\begin{theo} \mar{111} \label{111} There is an
isomorphism
\mar{mos074}\beq
\cJ^k(P)=\cJ^k\op\ot_{\cA^\bll-\cA} P, \qquad (a\ot_k bp)
\leftrightarrow (a\ot_1 b)\ot p. \label{mos074}
\eeq
\end{theo}

Then we have the $(\cA-\cA^\bll)$-module isomorphisms
\be
&&\hm_\cA(\cJ^k\ot P,Q)= \hm_{\cA-\cA^\bll} (P,\hm_\cA(\cJ^k,Q))=\\
&&\qquad
\hm_{\cA-\cA^\bll}(P,\dif_k(\cA,Q))=\dif_k(P,Q)=\hm_\cA(\cJ^k(P),Q).
\ee

Besides the monomorphism (\ref{5.44}):
\be
J^1: \cA\ni a\to \bb\otimes_1 a\in \cJ^1,
\ee
there exists an $\cA$-module monomorphism
\be
i_1: \cA \ni a  \to a\otimes_1 \bb\in \cJ^1.
\ee
With these monomorphisms, we have the canonical $\cA$-module
splitting
\mar{mos058}\beq
\cJ^1=i_1(\cA)\oplus \cO^1,  \qquad J^1(b)= \bb\ot_1 b=b\ot_1\bb +
(\bb\ot_1 b- b\ot_1\bb), \label{mos058}
\eeq
where the $\cA$-module $\cO^1$ is generated by elements $\bb\ot_1
b-b\ot_1 \bb$ for all $b\in\cA$. Let us consider a
$\cK$-homomorphism
\mar{mos045}\beq
d^1: \cA \ni b \to \bb\ot_1 b- b\ot_1\bb \in \cO^1. \label{mos045}
\eeq
This is a $\cO^1$-valued derivation of a $\cK$-ring $\cA$ which
obeys the Leibniz rule
\be
d^1(ab)= \bb\ot_1 ab-ab\ot_1\bb +a\ot_1 b  -a\ot_1 b  =ad^1b +
(d^1a)b.
\ee
It follows from the relation (\ref{5.53}) that $ad^1b=(d^1b)a$ for
all $a,b\in \cA$. Thus, seen as an $\cA$-module, $\cO^1$ is
generated by elements $d^1a$ for all $a\in\cA$.

Let $\cO^{1*}=\hm_{\cA}(\cO^1,\cA)$ be the dual  of an
$\cA$-module $\cO^1$. In view of the splittings (\ref{spr156}) and
(\ref{mos058}), the isomorphism (\ref{5.50}) leads to the duality
relation
\mar{5.81a}\beq
\gd\cA=\cO^{1*}, \qquad \gd\cA\ni u\leftrightarrow \f_u\in
\cO^{1*}, \qquad \f_u(d^1a)=u(a), \qquad  a\in \cA. \label{5.81a}
\eeq

\begin{ex} \mar{119} \label{119}
If $\cA=C^\infty(X)$ in Example \ref{115}, then $\cO^1=\cO^1(X)$
is a module of differential one-forms on $X$, and there is an
isomorphism $\cO^1(X)=T(X)^*$, besides the isomorphism
$T(X)=\cO^1(X)^*$ (\ref{5.81a}).
\end{ex}

Let us return to the first order jet module $\cJ^1(P)$ of an
$\cA$-module $P$. Due to the isomorphism (\ref{mos074}), the
isomorphism (\ref{mos058}) leads to the splitting
\mar{mos071}\ben
&& \cJ^1(P)= (\cA\oplus \cO^1)\op\ot_{\cA^\bll-\cA} P=
(\cA \op\ot_\cA P)\oplus (\cO^1\op\ot_{\cA^\bll-\cA} P), \label{mos071}\\
&& a\ot_1 bp\leftrightarrow  (ab +ad^1(b))\ot p. \nonumber
\een
Applying the epimorphism $\pi^1_0$ (\ref{+a13}) to this splitting,
one obtains the short exact sequence of $(\cA-\cA^\bll)$-modules
\mar{+175}\ben
&& 0\ar \cO^1\op\ot_{\cA^\bll-\cA} P\to \cJ^1(P)\ar^{\pi^1_0} P\ar 0, \label{+175}\\
&&  (a\ot_1 b -ab\ot_1 \bb)\ot p\to
(c\ot_1 \bb+ a\ot_1 b -ab\ot_1 \bb)\ot p \to cp. \nonumber
\een
It is canonically split by the $\cA^\bll$-homomorphism
\be
P\ni ap \to \bb\ot_1 ap= a\ot_1 p + d^1(a)\ot_1 p\in\cJ^1(P).
\ee
However, it need not be split by an $\cA$-homomorphism, unless $P$
is a projective $\cA$-module.

\begin{defi} \label{+176} \mar{+176}
A connection on an $\cA$-module $P$ is defined as an
$\cA$-homomorphism
\mar{+179'}\beq
\G:P\to \cJ^1(P), \qquad \G (ap)=a\G(p), \label{+179'}
\eeq
which splits the exact sequence (\ref{+175}).
\end{defi}

Given the splitting $\G$ (\ref{+179'}), let us define a
complementary morphism
\mar{+179} \beq
\nabla=J^1-\G: P\to \cO^1\op\ot_{\cA^\bll-\cA} P, \qquad
\nabla(p)= \bb\ot_1 p- \G(p).\label{+179}
\eeq
This also is called a connection though it in fact is a covariant
differential on a module $P$. This morphism satisfies the Leibniz
rule
\mar{+180}\beq
\nabla(ap)= d^1a \ot p +a\nabla(p), \label{+180}
\eeq
i.e., $\nabla$ is first order $(\cO^1\ot P)$-valued differential
operator on $P$. Thus, we come to the equivalent definition of a
connection \cite{kosz60}.

\begin{defi} \label{+181} \mar{+181}
A connection on an $\cA$-module $P$ is a $\cK$-homomorphism
$\nabla$ (\ref{+179}) which obeys the Leibniz rule (\ref{+180}).
\end{defi}

In view of the isomorphism (\ref{5.81a}), any connection in
Definition \ref{+181} determines a connection in the following
sense.

\begin{defi} \label{1016} \mar{1016}
A connection on an $\cA$-module $P$ is an $\cA$-homomorphism
\mar{1017}\beq
\nabla:\gd\cA\ni u\to \nabla_u\in \dif_1(P,P) \label{1017}
\eeq
such that, for each $u\in \gd\cA$, the first order differential
operator $\nabla_u$ obeys the Leibniz rule
\mar{1018}\beq
\nabla_u (ap)= u(a)p+ a\nabla_u(p), \quad a\in \cA, \quad p\in P.
\label{1018}
\eeq
\end{defi}

Definitions \ref{+181} and \ref{1016} are equivalent if
$\cO^1=\gd\cA^*$. For instance, this is the case of
$\cA=C^\infty(X)$ in Example \ref{119}.

In particular, let $P$ be a commutative $\cA$-algebra and $\gd P$
the derivation module of $P$ as a $\cK$-algebra. The $\gd P$ is
both a $P$- and $\cA$-modules. Then Definition \ref{1016} is
modified as follows.

\begin{defi} \label{mos088} \mar{mos088}
A connection on an $\cA$-algebra $P$ is an $\cA$-homomorphism
\mar{mos090}\beq
\nabla:\gd\cA\ni u\to \nabla_u\in \gd P\subset \dif_1(P,P),
\label{mos090}
\eeq
which is a connection on $P$ as an $\cA$-module, i.e., it obeys
the Leinbniz rule (\ref{1018}).
\end{defi}

For instance, if $P$ is an ideal of $\cA$, there is a unique
canonical connection $u\to\nabla_u=u$ on $P$.

\section{Differential operators on sections of a vector bundle}

Let $X$ be a smooth manifold which is customarily assumed to be
Hausdorff and second-countable (i.e., it has a countable base for
topology). Consequently, it has a locally compact space which is a
union of a countable number of compact subsets, a separable space,
a paracompact and completely regular space. Let $X$ be connected
and oriented.

Let $Y\to X$ be a vector bundle over $X$. Its global sections $s$
constitute a $C^\infty(X)$-module $Y(X)$.

Let $J^kY$ be a $k$-order jet manifold of $Y$ whose elements are
$k$-order jets of sections $s$ of $Y\to X$. It is a vector bundle
$J^kY\to X$ over $X$. There is a $C^\infty(X)$-module isomorphism
\mar{118}\beq
\cJ^k(Y(X))=J^kY(X) \label{118}
\eeq
of a $k$-order jet module $\cJ^k(Y(X))$ of $Y(X)$ and a module
$J^kY(X)$ of global sections of a $k$-order jet bundle $J^kY\to X$
of $Y\to X$ \cite{book09,sard09a}.

Let $E\to X$ be a vector bundle and $E(X)$ a $C^\infty(X)$-module
of global sections of $E$. By virtue of Theorem \ref{t6}, there is
the $C^\infty(X)$-module isomorphism (\ref{5.50}):
\mar{140}\beq
\dif_k(Y(X),E(X))=\hm_{C^\infty(X)}(\cJ^k(Y(X)),E(X)) \label{140}
\eeq
of the module $\dif_k(Y(X),E(X))$ of $k$-order $E(X)$-valued
differential operators on $Y(X)$ and the module
$\hm_{C^\infty(X)}(\cJ^k(Y(X)),E(X))$ of
$C^\infty(X)$-homomorphisms of $\cJ^k(Y(X))$ to $E(X)$. Since
$\cJ^k(Y(X))$ (\ref{118}) is a projective $C^\infty(X)$-module of
finite rank, the isomorphism (\ref{140}) takes the form
\mar{141}\ben
&& \dif_k(Y(X),E(X))=\hm_{C^\infty(X)}(\cJ^k(Y(X)),E(X))=
\label{141}\\
&& \qquad \cJ^k(Y(X))^*\op\ot_{C^\infty(X)} E(X)= ((J^kY)^*\ot
E)(X).\nonumber
\een
It follows that there is one-to-one correspondence between the
$k$-order $E(X)$-valued differential operators on $Y(X)$ and the
global sections of the vector bundle $(J^kY)^*\ot E$ where
$(J^kY)^*$ is the dual of a vector bundle $J^kY\to X$.

In particular, let $E=Y$. In accordance with Definition
\ref{1016}, a connection $\nabla$ on a module $Y(X)$ is a
$C^\infty(X)$-homomorphism
\mar{142}\beq
\nabla: \gd C^\infty(X)=T(X)\ni u\to \nabla_u\in \dif_1(Y(X),Y(X))
\label{142}
\eeq
such that, for each vector field $u$ on $X$, a first order
differential operator $\nabla_u$ obeys the Leibniz rule
(\ref{1018}):
\be
\nabla_u(fs)=(\bL_uf)s + f\nabla_us, \qquad s\in Y(X), \qquad f\in
C^\infty(X).
\ee
There is one-to-one correspondence between the connections
$\nabla^\G$ on a module $Y(X)$ and the linear connections $\G$ on
a vector bundle $Y\to X$ such that $\nabla^\G$ is the covariant
differential with respect to $\G$ \cite{book09,sard09a}.

For instance, let
\mar{150}\beq
R=X\times \Bbb R\to X \label{150}
\eeq
be a trivial bundle. Its global sections are smooth real functions
on $X$, i.e., $C^\infty(X)=R(X)$. A $k$-order jet manifold of this
bundle is diffeomorphic to a Whitney sum
\mar{113}\beq
J^kR= R\oplus T^*X\oplus\cdots \oplus \op\vee^k T^*X \label{113}
\eeq
of symmetric products of the cotangent bundle $T^*X$ to $X$. By
virtue of the isomorphism (\ref{118}), a $k$-order jet module
$\cJ^k(C^\infty(X))=\cJ^k(R(X))$ of a ring $C^\infty(X)$ is a
module of sections $J^kR(X)$ of the vector bundle $J^kR\to X$
(\ref{113}). Let $E\to X$ be a vector bundle. Then the module
$\dif_k(C^\infty(X),E(X))$ of $E(X)$-valued differential operators
on $C^\infty(X)$ is isomorphic to a module of sections of the
vector bundle $(J^kR)^*\ot E$ where
\mar{126}\beq
(J^kR)^*=R\oplus TX\oplus\cdots \oplus \op\vee^k TX \label{126}
\eeq
is the dual of the vector bundle (\ref{113}). Thus, we have
\mar{230}\beq
\dif_k(C^\infty(X),E(X))= ((J^kR)^*\ot E)(X). \label{230}
\eeq

In particular, let $E=R$. Then the module of $k$-order
$C^\infty(X)$-valued
 differential operators on $C^\infty(X)$ is
\mar{231}\beq
\dif_k(C^\infty(X)= (J^kR)^*(X). \label{231}
\eeq
For instance, the module $\dif_1(C^\infty(X))$ of first order
differential operators on $C^\infty(X)$ is isomorphic to
$C^\infty(X) \oplus T(X)$ in accordance with the decomposition
(\ref{spr156}).

\section{Differential operators on sections with compact support}

Let us consider a $C^\infty(X)$-module $\cD(Y)\subset Y(X)$ of
sections with compact support of a vector bundle $Y\to X$. It is
endowed with the following topology \cite{cirr}.

Let $J^\infty Y$ be the topological inductive limit of $J^kY$,
$k\in \Bbb N$ which is a Fr\'echet (not smooth) manifold
\cite{book09,sard09a}. It is a topological vector bundle
\be
\pi^\infty_0:J^\infty Y\to X.
\ee
There is a certain class $\cQ^0_\infty Y$ of real functions on
$J^\infty Y$ called the smooth functions on $J^\infty Y$. Given a
function $f\in \cQ^0_\infty Y$ and a point $z\in J^\infty Y$,
there exists an open neighborhood $U$ of $z$ such that $f|_U$ is
the pull-back of a smooth function on some finite order jet
manifold $J^kY$.

Let $\cF_\infty Y\subset \cQ^0_\infty Y$ denote a subset of smooth
functions $\f$ on $J^\infty Y$ which are of finite jet order
$[\f(K)]$ on a subset $(\pi^\infty_0)^{-1}(K)\subset J^\infty Y$
over any compact subset $K\subset X$, and which are linear on
fibres of $J^\infty Y\to X$. With $\f\in \cF_\infty Y$, one can
define a seminorm
\mar{101}\beq
p_\f(s)=\op\sup_{x\in X}|J^\infty s^*\f| \label{101}
\eeq
on $\cD(Y)$ where $J^\infty s$ denotes the jet prolongation of a
section $s$ to a section of $J^\infty Y\to X$. The seminorm
(\ref{101}) is well defined because
\be
(J^\infty s^*\f)= (J^{[\f({\rm supp}\,s)]} s^*\f)(x)
\ee
is a smooth function with compact support on $X$.

The set of seminorms $p_\f(s)$ (\ref{101}) for all functions
$\f\in\cF_\infty Y$ on $J^\infty Y$ yields a locally convex
topology on $\cD(Y)$ called the $LF$-topology. With this topology,
$\cD(Y)$ is a nuclear complete reflexive vector space, an
inductive limit of a countable family of separable Fr\'echet
spaces \cite{trev}.

If $Y=R$, the $\cD(R)=\cD(X)$ is the well-known nuclear space of
test functions on a manifold $X$. There is a $C^\infty(X)$-module
isomorphism
\mar{102}\beq
\cD(Y)=Y(X)\op\otimes_{C^\infty(X)} \cD(X)\subset
Y(X)\op\otimes_{C^\infty(X)} Y(X)=Y(X). \label{102}
\eeq

\begin{lem} \mar{104} \label{104}
Any global section $\si$ of the dual vector bundle $Y^*\to X$
yields a continuous homomorphism of the topological vector spaces
\mar{105}\beq
\si: \cD(Y)\ni s \to (s,\si)\in \cD(X). \label{105}
\eeq
\end{lem}

\begin{proof}
Let $p_F$ (\ref{101}) be a seminorm on $\cD(X)$ where
$F\in\cF_\infty R$. A global section $\si$ of $Y^*\to X$ defines a
bundle morphism $\si: Y\to R$ over $X$ possessing  a jet
prolongation
\be
J^\infty\si: J^\infty Y\ar_X J^\infty R.
\ee
Then we have a smooth real function
\mar{153}\ben
&& F_\si=F\circ J^\infty\si: J^\infty Y\ar_X J^\infty
R\to \Bbb R, \label{153}\\
&& F_\si(z)=(F\circ J^{[F(\pi^\infty_0(z))]}\si)(z), \qquad z\in J^\infty
Y,\nonumber
\een
on $J^\infty Y$ which belongs to $\cF_\infty Y$. The function
$F_\si$ (\ref{153}) yields the seminorm $p_{F_\si}$ (\ref{101}) on
$\cD(Y)$ such that
\be
p_F((s,\si))=p_{F_\si}(s)
\ee
for all $s\in \cD(Y)$. It follows that $p_F((s,\si))<\ve$ iff
$p_{F_\si}(s)<\ve$ and, consequently, $\si$ (\ref{105}) is an open
continuous map.
\end{proof}

In the case of $Y=R$, the homomorphism (\ref{105}) is a
multiplication
\be
g: \cD(X)\ni f\to gf \in\cD(X), \qquad g\in C^\infty(X),
\ee
which thus is an open continuous map. Consequently, $\cD(X)$ is a
topological $C^\infty(X)$-algebra.

Given a section $\si$ of $Y^*\to X$ and a function $f\in
C^\infty(X)$, let consider the homomorphism (\ref{105}):
\be
(f\si) (s)=f\si(s)= \si(fs), \qquad s\in \cD(Y).
\ee
Since the morphisms $f\si$ and $\si$ for any $\si\in Y^*(X)$ are
continuous and open, the multiplication $s\to fs$ also is an open
continuous homomorphism of $\cD(Y)$. It follows that $\cD(Y)$ is a
topological $C^\infty(X)$-module. Accordingly, the homomorphism
(\ref{105}) is a continuous $C^\infty(X)$-homomorphism.

Let $\cJ^k(\cD(Y))$ be a $k$-order jet module of a
$C^\infty(X)$-module $\cD(Y)$. It is a submodule of $\cJ^k(Y(X))$
and, due to the isomorphisms (\ref{118}) and (\ref{102}), a
submodule
\mar{250}\beq
\cJ^k(\cD(Y))= \cD(J^k(X))\subset J^kY(X) \label{250}
\eeq
of sections with compact support of the jet bundle $J^kY\to X$.

\begin{theo} \mar{130} \label{130} Let $E\to X$ be a vector bundle. There is one-to-one
correspondence between $k$-order $E(X)$-valued  differential
operators on sections and sections with compact support of $Y\to
X$.
\end{theo}

\begin{proof}
Of course, any differential operator on $Y(X)$ also is that on
$\cD(Y)$. Let $\Delta$ be a $k$-order $E(X)$-valued  differential
operator on a $C^\infty(X)$-module $\cD(Y)$. By virtue of Theorem
\ref{t6} and the isomorphism (\ref{250}), it defines a unique
$C^\infty(X)$-homomorphism
\be
{\got f}^\Delta: \cJ^k(\cD(Y))=\cD(J^k(X))\to E(X)
\ee
of sections with compact support of a vector jet bundle $J^kY\to
X$ to $E(X)$. Let us show that this morphism is extended to an
arbitrary global section $s$ of $J^kY\to X$. A smooth manifold $X$
admits an atlas whose cover consists of a countable set of open
subsets $U_i$ such that that their closures $\ol U_i$ are compact
\cite{greub}. Let $\{f_i\}$ be a subordinate partition of unity,
where each $f_i$ is a smooth function with a support
supp$f_i\subset U_i\subset \ol U_i$, i.e., with compact support.
Each point $x\in X$ has an open neighborhood which intersects only
a finite number of supp$f_i$, and
\be
\op\sum_i f_i(x)=1, \qquad x\in X.
\ee
Then one can put
\be
s=\op\sum_i f_is, \qquad f_is\in \cD(J^kY)=\cJ^k(\cD(Y)),
\ee
and define a $C^\infty(X)$-homomorphism
\be
{\got f}^\Delta (s)= \op\sum_i {\got f}^\Delta(f_is)
\ee
of $J^kY(X)=\cJ^k(Y(X))$ to $E(X)$. By virtue of Theorem \ref{t6},
it provides a $k$-order $E(X)$-valued differential operator on
$Y(X)$.
\end{proof}

It follows from Theorem \ref{130} and the isomorphism (\ref{141})
that
\mar{131}\beq
\dif_k(\cD(Y),E(X))=\dif_k(Y(X),E(X))=((J^kY)^*\ot E)(X).
\label{131}
\eeq

\begin{lem} \mar{190} \label{190} Any $E(X)$-valued differential
operator on $\cD(Y)$ is $\cD(E)$-valued.
\end{lem}

\begin{proof} By virtue of Theorem \ref{t6}, any $k$-order $E(X)$-valued
 differential operator $\Delta$ on $\cD(Y)$ factorizes
through a $C^\infty(X)$-homomorphism
\be
{\got f}^\Delta: \cJ^k(\cD(Y))\to E(X),
\ee
which is $\cD(E)$-valued due to the isomorphism (\ref{250}).
\end{proof}

In particular, let $E=Y$. Then $Y(X)$-valued differential
operators on $\cD(Y)$ are $\cD(Y)$-valued.

Let $Y=R$ and $\cD(Y)=\cD(X)$ a $C^\infty(X)$-algebra of test
functions on $X$. By virtue of Theorem \ref{130}, there is
one-to-one correspondence between $E(X)$-valued differential
operators on $\cD(X)$ and $C^\infty(X)$. Then it follows from the
isomorphism (\ref{230}) that
\be
\dif_k(\cD(X),E(X))= ((J^kR)^*\ot E)(X).
\ee
In particular, the derivations (\ref{116}) of $C^\infty(X)$ also
are derivations of an $\Bbb R$-algebra $\cD(X)$. Since $\cD(X)$ is
an ideal of $C^\infty(X)$, there exists a unique canonical
connection $\nabla_u=u$ on $\cD(X)$.

\begin{theo} \mar{106} \label{106}
Any $\cD(Y)$-valued differential operator on a
$C^\infty(X)$-module $\cD(Y)$ is continuous.
\end{theo}

\begin{proof}
Let $p_\f$ (\ref{101}) be a seminorm on $\cD(X)$ where
$F\f\in\cF_\infty Y$. By virtue of Theorem \ref{t6}, any $k$-order
$Y(X)$-valued  differential operator $\Delta$ on $\cD(Y)$ yields a
bundle morphism
\be
\ol{\got f}^\Delta: J^k Y\ar_X Y
\ee
such that, given a section $s$ of $Y\to X$, we have
\be
\Delta(s)=\ol{\got f}\circ J^ks.
\ee
Let us consider its jet prolongations
\be
J^r\ol{\got f}^\Delta: J^{r+k} Y\ar_X J^r Y, \qquad
J^\infty\ol{\got f}^\Delta: J^\infty Y\ar_X J^\infty Y.
\ee
Then we have a smooth real function
\mar{195}\ben
&& \phi_\Delta=\phi\circ J^\infty\ol{\got f}^\Delta: J^\infty Y\ar_X J^\infty
Y\to \Bbb R, \label{195}\\
&& \phi_\Delta(z)=(\phi\circ J^{[\phi(\pi^\infty_0(z))]}\ol{\got f}^\Delta)(z),
\qquad z\in J^\infty Y,\nonumber
\een
on $J^\infty Y$ which belongs to $\cF_\infty Y$. The function
$\phi_\Delta$ (\ref{195}) yields the seminorm $p_{\phi_\Delta}$
(\ref{101}) on $\cD(Y)$. Then we have
\be
p_\phi(\Delta(s))=p_{\phi_\Delta}(s).
\ee
It follows that $p_\phi(\Delta(s))<\ve$ iff
$p_{\phi_\Delta}(s)<\ve$ and, consequently, $\Delta$ is an open
continuous map.
\end{proof}

\section{Differential operators on Schwartz distributions}

Given a $LF$-space $\cD(Y)$ of sections with compact support of a
vector bundle $Y\to X$, let $\cD(Y)'$ be the topological dual of
$\cD(Y)$. Its elements are continuous forms
\be
\cD(Y)'\ni\psi: \cD(Y)\ni s\to <s,\psi>\in\Bbb R
\ee
on $\cD(Y)$ called the Schwartz distributions. The vector space
$\cD(Y)'$ is provided with the strong topology (which coincides
with all topologies of uniform converges). It is nuclear, and the
topological dual of $\cD(Y)'$ is $\cD(Y)$. A $LF$-topology of
$\cD(Y)$ also coincides with all topologies of uniform converges
\cite{trev}.

For instance, let $Y=R$ (\ref{150}). Then $\cD(Y)=\cD(X)$ is the
space of test functions on a manifold $X$ and its topological dual
$\cD(X)'$ is the familiar space of Schwartz distributions on test
functions on a manifold $X$. Since a $LF$-topology is finer than
the topology on $\cD(X)$ induced by the inductive limit topology
of the space $K(X)$ of continuous real functions on $X$, any
measure on $X$ exemplifies a Schwartz distribution. For instance,
any density
\be
L\in \op\w^n T^*X(X),  \qquad n=\di X,
\ee
on an oriented manifold $X$ is a Schwartz distribution on
$\cD(X)$.

Let $\Delta$ be a continuous homomorphism of a topological vector
space $\cD(Y)$. Then
\mar{158}\beq
\Delta':\cD(Y)'\ni \psi\to \psi\circ \Delta\in\cD(Y)', \label{158}
\eeq
is a morphism of the topological dual $\cD(Y)'$ of $\cD(Y)$. It is
called the transpose or the dual of $\Delta$. We have
\be
<\Delta (s),\psi>=<s,\Delta'(\psi)>.
\ee
Since topologies on $\cD(Y)$ and $\cD(Y)'$ coincide with the weak
ones, the transpose operator $\Delta'$ (\ref{158}) is continuous.

For instance, the transpose (\ref{158}) of the multiplication
$s\to fs$, $f\in C^\infty(X)$, in $\cD(Y)$ is the multiplication
\mar{202}\beq
\psi\to f\psi, \qquad <fs,\psi>=<s,f\psi> \label{202}
\eeq
which makes $\cD(Y)'$ into a topological $C^\infty(X)$-module. In
particular, $\cD(X)'$ also is a $C^\infty(X)$-module.

Let $\si$ be a global section of the dual bundle $Y^*\to X$. By
virtue of Lemma \ref{104}, it defines the continuous
$C^\infty(X)$-homomorphism (\ref{105}) of $\cD(Y)$ to $\cD(X)$
and, accordingly, the dual $C^\infty(X)$-homomorphism
\mar{155}\beq
\si': \cD(X)'\ni\xi\to \xi\circ\si\in \cD(Y)'. \label{155}
\eeq
It follows that there is a $C^\infty(X)$-module monomorphism
\mar{170}\beq
Y^*(X)\op\ot_{C^\infty(X)} \cD(X)'\to \cD(Y)' \label{170}
\eeq
such that
\be
(\si\ot \xi)(s)=\xi((s,\si)), \qquad s \in \cD(Y), \qquad \si\in
Y^*(X), \qquad \xi\in \cD(X)'.
\ee

\begin{theo} \mar{171} \label{171}
The monomorphism (\ref{155}) is a $C^\infty(X)$-module isomorphism
\mar{172}\beq
\cD(Y)'=Y^*(X)\op\ot_{C^\infty(X)} \cD(X)'. \label{172}
\eeq
\end{theo}

\begin{proof}
A vector bundle $Y\to X$ of fibre dimension $m$ admits a finite
atlas $\{(U_i,h_i),\rho_{ij}\}$, $i,j=1,\ldots,k$, \cite{greub}.
Given a smooth partition of unity $\{f_i\}$ subordinate to a cover
$\{U_i\}$, let us put
\be
l_i=f_i(f_1^2+\cdots +f_k^2)^{-1/2}.
\ee
It is readily observed that $\{l_i^2\}$ also is a partition of
unity subordinate to $\{U_i\}$. Then any section $s\in \cD(Y)$ is
represented by a tuple $(s_1,\ldots,s_k)$ of local $\Bbb
R^m$-valued functions $s_i=h_i\circ s|_{U_i}$ which fulfil the
relations
\mar{spr709}\beq
s_i=\op\sum_j \rho_{ij}(s_j)l^2_j. \label{spr709}
\eeq
Let us consider a topological vector space
\mar{175}\beq
\op\oplus^{mk}\cD(X), \label{175}
\eeq
which also is a topological $C^\infty(X)$-module. There are both a
continuous $C^\infty(X)$-monomorphism
\be
\g:\cD(Y)\ni s \to (l_1s_1,\ldots,l_ks_k)\in \op\oplus^{mk}\cD(X)
\ee
and a continuous $C^\infty(X)$-epimorphism
\mar{177}\beq
\Phi:\op\oplus^{mk}\cD(X) \ni (t_1,\ldots,t_k)\to (\wt
s_1,\ldots,\wt s_k)\in \cD(Y), \qquad \wt s_i=\op\sum_j
\rho_{ij}(l_jt_j). \label{177}
\eeq
In view of the relations (\ref{spr709}),
\be
\Phi\circ \g=\id \cD(Y),
\ee
and we have a decomposition
\mar{174}\ben
&& \op\oplus^{mk}\cD(X)=\g(\cD(Y)) \oplus \Ker \Phi, \label{174}\\
&& t_i=[l_i\op\sum_j\rho_{ij}(l_jt_j)] + [t_i - l_i\op\sum_j\rho_{ij}(l_jt_j)], \nonumber
\een
where $\g(\cD(Y))$ consists of elements $(t_i)$ satisfying the
condition
\mar{179}\beq
t_i = l_i\op\sum_j\rho_{ij}(l_jt_j). \label{179}
\eeq
The topological dual of the topological vector space (\ref{175})
is a $C^\infty(X)$-module
\mar{185}\beq
\op\oplus^{mk}\cD(X)' \label{185}
\eeq
with elements $(\ol t_1,\ldots,\ol t_k)$. The epimorphism $\Phi$
(\ref{177}) yields a $C^\infty(X)$-monomorphism
\be
\Phi':\cD(Y)'\ni \xi \to \xi\circ\Phi \in \op\oplus^{mk}\cD(X)'
\label{178}
\ee
such that $\Phi'(\cD(Y)')$ vanishes on $\Ker \Phi$ in the
decomposition (\ref{174}). To describe $\Phi'(\cD(Y)')$, let us
consider the dual vector bundle $Y^*\to X$ provided with the
conjugate atlas $\{(U_i,\ol h_i),\ol\rho_{ij}\}$ such that, for
arbitrary sections $s$ of $Y\to X$ and $\si$ of $Y^*\to X$, the
equality
\mar{180}\beq
(h_i\circ s, \ol h_i\circ \si)|_{U_i\cap U_j}=(\rho_{ij}\circ
h_j\circ s, \ol\rho_{ij}\circ \ol h_j\circ \ol s) = (h_j\circ s,
\ol h_j\circ \si)_{U_i\cap U_j} \label{180}
\eeq
holds. Then it is readily verified that the image $\Phi'(\cD(Y)')$
of $\cD(Y)'$ in the $C^\infty(X)$-module (\ref{185}) consists of
elements $(\ol t_i)$ satisfying the condition
\mar{189}\beq
\ol t_i = l_i\op\sum_j\ol\rho_{ij}(l_jt_j) \label{189}
\eeq
(cf. the condition (\ref{179})). This fact leads to the
isomorphism (\ref{172}).
\end{proof}

Since Schwartz distributions on sections with compact support of a
vector bundle $Y\to X$ constitute a $C^\infty(X)$ module
$\cD(Y)'$, differential operators on them can be introduced in
accordance with Definition \ref{ws131}.

We restrict our consideration to $\cD(Y)'$-valued differential
operators on $\cD(Y)'$. Of course, any multiplication (\ref{202})
is a zero-order differential operator on $\cD(Y)'$.

In accordance with Theorem \ref{106}, any $\cD(Y)$-valued
differential operator $\Delta$ on a $C^\infty(X)$-module $\cD(Y)$
of sections with compact support is a continuous morphism of a
topological vector space $\cD(Y)$. Then it defines the dual
morphism $\Delta'$ (\ref{158}) of $\cD(Y)'$.

\begin{theo} \mar{205} \label{205} The transpose $\Delta'$
(\ref{158}) of a $k$-order differential operator $\Delta$ on
$\cD(Y)$ is a differential operator on $\cD(Y)'$ in accordance
with Definition \ref{ws131}.
\end{theo}

\begin{proof} The proof is based on the fact that $\dl_f\Delta'$,
$f\in C^\infty(X)$, (\ref{spr172}) is the transpose of
$-\dl_f\Delta$.
\end{proof}

For instance, any connection $\nabla$ (\ref{142}) on $Y(X)$ and,
consequently, on $\cD(Y)$ define the transpose $\nabla_u'$ on
$\cD(Y)'$ for any vector field $u$ on $X$. We have
\be
&& <\f, \nabla_u'(f\psi)>= <f\nabla_u(\f),\psi>=<\nabla(f\f),\psi>
-<\bL_u(f)\f,\psi>= \\
&& \qquad <\f,f\nabla_u'(\psi)> +<\f,\bL_{-u}(f)\psi>.
\ee
A glance at this equality shows that $-\nabla_u'$ is a connection
on $\cD(Y)'$.

In particular, let $Y=R$, and let $\bL_u$, $u\in TX$, be the
derivation (\ref{116}) of $\cD(X)$. Its transpose $\bL_u'$ is
called the Lie derivative of Schwartz distributions $\psi\in
\cD(X)'$ along $u$. In particular, if
\be
\cD(X)'\ni \psi=\ol\psi d^nx
\ee
is a density on $X$, then
\be
\bL_u'(\psi)=\bL_{-u}(\psi)=-d(u\rfloor\psi)=-\dr_\m(u^\m\ol\psi)d^nx.
\ee
It is a derivation of $\cD(X)'$ because
\be
\bL_u'(f\psi)=\bL_{-u}(f\psi)=\bL_{-u}(f)\psi +f\bL_{-u}(\psi).
\ee

The transpose $\Delta'$ (\ref{158}) on $\cD(Y)'$ of a differential
operator $\Delta$ on $\cD(Y)$ is continuous.

However, a differential operator on $\cD(Y)'$ need not be the
transpose of a differential operator on $\cD(Y)$. Since $\cD(Y)$
is reflexive and topologies on $\cD(Y)$ and $\cD(Y)'$ coincide
with the weak ones, one can show the following.

\begin{theo} \mar{220} \label{220} A differential operator on
$\cD(Y)'$ is the transpose of a differential operator on $\cD(Y)$
iff it is continuous.
\end{theo}

It follows that there is one-to-one correspondence between
continuous differential operators on $\cD(Y)'$ and differential
operators on $\cD(Y)$ whose module is isomorphic to
\mar{245}\beq
\dif_k(\cD(Y)) = ((J^kY)^*\op\ot_X Y)(X) \label{245}
\eeq
in accordance with the isomorphism (\ref{131}).

Jet modules $\cJ^k(\cD(Y)')$ of a $C^\infty(X)$-module $\cD(Y)'$
are introduced in accordance with Definition \ref{200}. By virtue
of Theorem \ref{t6}, any $k$-order differential operator $\Phi$ on
$\cD(Y)'$ is represented by a $C^\infty(X)$-homomorphism
\mar{303}\beq
{\got f}^\Phi: \cJ^k(\cD(Y)')\to \cD(Y)'. \label{303}
\eeq

Let $Y^*$ be the dual of a vector bundle $Y\to X$. Then
$Y^*(X)\subset \cD(Y)'$ is a $C^\infty(X)$-submodule of $\cD(Y)'$.
It is easily verified that, if $\Delta$ is a differential operator
on $\cD(Y)$, its transpose restricted to $Y^*(X)$ takes its values
into $Y^*(X)$. Therefore a differential operator $\Phi$ on
$\cD(Y)'$ is the transpose of a differential operator on $\cD(Y)$
only if
\mar{334}\beq
{\got f}^\Phi(J^kY^*(X))\subset Y^*(X), \label{334}
\eeq
i.e., the restriction
\mar{304}\beq
\Phi_*=\Phi|_{Y^*(X)} \label{304}
\eeq
of $\Phi$ is a $Y^*(X)$-valued differential operator on $Y^*(X)$.
Let this condition hold. Then $\Phi$ is the transpose of a
differential operator on $\cD(Y)$ if it is the double transpose of
$\Phi_*$ (\ref{304}) as follows. Let $\Phi_*$ (\ref{304}) be a
$Y^*(X)$-valued differential operator on $Y^*(X)\subset \cD(Y)'$
and, consequently, on a $C^\infty(X)$-module $\cD(Y^*)$. Let
$\Phi_*'$ be its transpose on $\cD(Y^*)'$. Then the restriction of
$\Phi_*'$ to $\cD(Y)\subset \cD(Y^*)'$ is a $\cD(Y)$-valued
differential operator on $\cD(Y)$. Let $\Phi_*''$ be its transpose
on $\cD(Y)'$. If $\Phi\neq \Phi_*''$, then $\Phi-\Phi_*''$ is a
non-zero differential operator vanishing on $\cD(Y^*)$. Such a
differential operator fails to be the transpose of a differential
operator on $\cD(Y)$ which must be neither non-zero because
$\Phi-\Phi_*''$ vanishes on $\cD(Y^*)$ nor zero because
$\Phi-\Phi_*''$ is non-zero.

\begin{ex}
Let us consider a space $\cD(X)$ of test functions on a manifold
$X$. A vector space of Schwartz distributions $\cD(X)'$ is a free
$\Bbb R$-module whose basis $B$ contains the Dirac measure $\e_x$
at a point $x\in X$. Let us consider a vector subspace
$\cD(X)'_x\subset \cD(X)'$ whose basis is $B\setminus \e_x$. A
vector subspace $\cD(X)'_x$ also is a $C^\infty(X)$-submodule of
$\cD(X)'$ because $f\e_x=f(x)\e_x$, $f\in C^\infty(X)$.
Accordingly, an $\Bbb R$-epimorphism
\be
\g:\cD(X)'\to \cD(X)'_x,\qquad \g(\cD(X)'_x)=\id \cD(X)'_x, \qquad
\g(\e_x)=0,
\ee
also is a $C^\infty(X)$-homomorphism, i.e., a zero order
differential operator on $\cD(X)'$. It is not the transpose of a
differential operator on $\cD(X)$.
\end{ex}

\end{document}